\documentclass{elsarticle}
\usepackage{amssymb}
\usepackage{amsmath}
\usepackage{amsthm}
\usepackage[utf8]{inputenc}
\usepackage[english]{babel}
\usepackage[
  pdftex,
  bookmarks=true,
  bookmarksnumbered=true,
  hypertexnames=false,
  breaklinks=true,
  linkbordercolor={0 0 1}
  ]{hyperref}
\usepackage{algorithm}
\usepackage{algpseudocode}
\usepackage{color,todonotes}
\usepackage{mathdots}

\newcommand{\R}{\mathbb{R}}

\newcommand{\N}{\mathbb{N}}

\newcommand{\Q}{\mathbb{Q}}

\newcommand{\Z}{\mathbb{Z}}

\newcommand{\A}{\mathbb{A}}
\newcommand{\infnorm}[1]{\left\lVert{#1}\right\rVert}
\newcommand{\sigmap}[1]{{\Sigma{#1}}}
\newcommand{\degp}[1]{\operatorname{deg}(#1)}
\newcommand{\taylor}[3]{{T_{#1}^{#2}#3}}
\newcommand{\intinterv}[2]{\{#1,\ldots,#2\}}
\newcommand{\bigO}[1]{\mathcal{O}\left(#1\right)}
\newcommand{\softO}[1]{\tilde{\mathcal{O}}\left(#1\right)}
\newcommand{\IntI}{\operatorname{Int}}
\newcommand{\LenI}{\operatorname{Len}}

\newcommand{\poly}{\operatorname{poly}}

\newcommand{\Rsize}[1]{\mathfrak{L}\left(#1\right)}

\newcommand{\lemref}[1]{Lemma~\ref{#1}}
\newcommand{\thref}[1]{Theorem~\ref{#1}}

\newcommand{\secref}[1]{Section~\ref{#1}}
\newcommand{\algoref}[1]{Algorithm~\ref{#1}}

\newcommand{\algassign}{:=}

\newcommand\numberthis{\addtocounter{equation}{1}\tag{\theequation}}

\numberwithin{figure}{section}
\newtheorem{definition}{Definition}
\newtheorem{example}[definition]{Example}

\newtheorem{theorem}[definition]{Theorem}

\newtheorem{lemma}[definition]{Lemma}

\makeatletter
\let\c@algorithm\c@definition
\makeatother

\renewcommand{\theequation}{\arabic{equation}}

\algnewcommand\Comments[1]{\item[\textbf{Comments\ifthenelse{\equal{#1}{}}{}{ (#1)}}]}
\algnewcommand\CommentLine[2]{\item[\textbullet\hspace{.5em}\textbf{Line #1:}] #2}
\algnewcommand\Definition[2][]{\item[\textbf{Definition\ifthenelse{\equal{#1}{}}{}{ (#1)}:}] #2}

\journal{Theoretical Computer Science}

\begin{document}

\begin{frontmatter}

\title{Computational complexity of solving polynomial
differential equations over unbounded domains}

\author[fct,lix]{Amaury Pouly\corref{cor1}}
\ead{pamaury@lix.polytechnique.fr}

\author[fct,sqig]{Daniel S. Gra\c{c}a}
\ead{dgraca@ualg.pt}

\cortext[cor1]{Corresponding author.}

\address[lix]{LIX -
1 rue Honoré d'Estienne d'Orves,
Bâtiment Alan Turing,
Campus de l'École Polytechnique,
91120 Palaiseau,
France. Telephone: (+33)177578015}
\address[fct]{FCT da Universidade do Algarve,
Campus de Gambelas,
8005-139 Faro,
Portugal.
Telephone: (+351)289800900 (extension 7663), Fax: (+351)289800066}
\address[sqig]{SQIG/Instituto de Telecomunica\c{c}\~{o}es, Lisbon, Portugal}

\begin{abstract}
In this paper we investigate the computational complexity of solving ordinary
differential equations (ODEs) $y^{\prime}=p(y)$ over \emph{unbounded time domains},
where $p$ is a vector of polynomials.
Contrarily to the bounded (compact) time case, this problem has not been well-studied,
apparently due to the ``intuition'' that it can always be reduced
to the bounded case by using rescaling techniques. However, as we show in this paper,
rescaling techniques do not seem to provide meaningful insights on the complexity of
this problem, since the use of such techniques introduces a dependence on parameters
which are hard to compute.

We present algorithms which numerically solve these ODEs over unbounded time domains.
These algorithms have guaranteed accuracy, i.e.~given some arbitrarily large time $t$ 
and error bound $\varepsilon$ as input,
they will output a value $\tilde{y}$ which satisfies $\|y(t)-\tilde{y}\|\leq\varepsilon$.
We analyze the complexity of these algorithms and show that they compute $\tilde{y}$ in time polynomial
in several quantities including the time $t$, the accuracy of the output $\varepsilon$
and the length of the curve $y$ from $0$ to $t$, assuming it exists until time $t$.
We consider both algebraic complexity and bit complexity.
\end{abstract}


\end{frontmatter}

\section{Introduction}\label{sec:solve:intro}

The purpose of this paper is to characterize the computational complexity needed to solve a
polynomial initial-value problem (PIVP) defined by
\begin{equation}
\left\{
\begin{array}{@{}r@{}l}
y^{\prime}(t)&=p(y(t)) \\
y(t_0)&=y_0
\end{array}
\right.\label{eq:pivp}
\end{equation}
over an unbounded time domain. Since the system is autonomous, we can assume, without loss of generality, that $t_0=0$.
More precisely, we want to compute $y(t)$ with accuracy $2^{-n}$,
where $t\in\R$, $n\in\N$, and a description of $p$ are given as inputs, and $y$ is the solution of \eqref{eq:pivp}.
We have to assume the existence of $y$ until time $t$ because this problem is
undecidable, even for polynomial ODEs \cite{GBC07}.

\subsection{Why polynomial differential equations?}

In this paper we study the computational complexity of solving initial-value problems (IVPs)
$y^{\prime}=f(t,y), y(t_0)=y_0$, where $f$ is a vector of polynomials,
over (potentially) unbounded domains. The reader may ask:
``why do you restrict $f$ to polynomials when there are several results about
the computational complexity of solving IVPs for the more general case where $f$ is Lipschitz?''.
There are, indeed, several results (see Section \ref{sec:solve:related_work} for some references)
which analyze the computational complexity of solving Lipschitz IVPs in \emph{bounded domains}.
And, in \emph{bounded domains} polynomials are Lipschitz functions (since they are of class $C^1$)
and therefore those above-mentioned results also apply to PIVPs.

However, in this paper we tackle the problem of computing the solutions of IVPs over
\emph{unbounded domains} and in that respect the previous results do not apply,
and no easy technique seems to establish a bridge between the bounded and unbounded case
(some authors informally suggested us that a ``rescalling technique'' could be used,
but this technique does not work, as we will see in Section \ref{sec:rescalling}).
In some sense, the unbounded case is more general than the bounded case:
if you know the complexity needed to solve an IVP over, e.g~$\R$,
then you can easily restrict this general case to give a bound for the complexity needed
to solve the same IVP over, e.g.~ $[0,1]$, but the reverse is not evident.
For this reason, it seems natural that results about the computational complexity
of IVPs over unbounded domains should be harder to get (or at least should not be easier to get)
than similar results for the bounded case.

That's the first reason why we use PIVPs: they are not trivial (polynomials do not
satisfy a Lipschitz condition over an unbounded domain, contrarily to simpler functions like linear functions)
but yet have ``nice'' properties which we can exploit to deal with the
harder case of establishing the computational complexity of solving IVPs over unbounded domains.

The second reason to use PIVPs is that they include a fairly broad class of IVPs,
since any IVP written with the usual functions of Analysis (trigonometric functions,
exponentials, their composition and inverse functions, etc.) can be rewritten as PIVPs,
as shown in \cite{WWSPC06}, \cite{GCB08}.

\subsection{A note on rescaling}\label{sec:rescalling}

It is tempting to think that the unbounded time domain case can be reduced to the bounded time one.
We would like to note that this not the case unless the bounded time case complexity
is studied in terms of \emph{all parameters} which is never the case. Indeed a very simple
example illustrates this problem. Assume that $y:I\rightarrow\R^d$ satisfies the following system:
\[
\left\{\begin{array}{@{}r@{}l}y_1(0)&=1\\y_2(0)&=1\\\ldots&\\ y_n(0)&=1
\end{array}\right.
\qquad
\left\{\begin{array}{@{}r@{}l}y_1'(t)&=y_1(t)\\y_2'(t)&=y_1(t)y_2(t)\\\ldots&\\
y_d'(t)&=y_1(t)\cdots y_n(t)
\end{array}\right.
\]
Results from  the literature (namely \cite{MM93} -- see Section \ref{Sec:PreliminaryResults})
show that for any fixed, compact $I$, $y$ is polynomial time (precision-)computable
(i.e.~for any $t\in I$ we can compute an approximation of $y(t)$ with precision $2^{-n}$
in time polynomial in $n$ -- see e.g.~\cite{BHW08}).
On the other hand, this system can be solved explicitly and yields:
\[y_1(t)=e^t\qquad y_{n+1}(t)=e^{y_n(t)-1}
\qquad y_d(t)=e^{e^{\iddots^{e^{e^t}-1}}-1}\]
One immediately sees that $y_d$ being a tower of exponentials prevents $y$ from being polynomial
time (precision-)computable over $\R$, for any reasonable notion, although $y_d$ (and $y$) is polynomial time (precision-)computable over any \emph{fixed} compact.

This example clearly shows that the solution of an IVP (or even of a PIVP) can be polynomial time computable on any fixed compact, while it may not necessarily be polynomial time computable over $\R$.
In fact this example provides
an even stronger counter-example: the discrepancy between the bounded and unbounded
time domain can be arbitrarily high. Note however that this discrepancy arises
because in the bounded time case, the size of the compact $I$ is not taken as a
parameter of the problem (because it is fixed). Also note that the dimension $d$
of the system is hardly ever taken into account, although it has a huge influence
on the resulting complexity. More precisely, if $I$ is bounded
then the complexity of computing $y(t)$ can be seen to be polynomial in $t$, but more than exponential
in $|I|$ and $d$: this part is usually hidden in the ``big-O'' part of the constants in the function measuring the complexity for the bounded case.

\subsection{Contributions}
 In this paper we give several contributions to the problem of solving the polynomial
initial value problems \eqref{eq:pivp}. The main result of this paper is given by \thref{th:solve_pivp_ex_taylor}.
Namely we present an algorithm which solves \eqref{eq:pivp} and
\begin{itemize}
\item show that our algorithm works even for unbounded time interval $I$,
\item we analyze the complexity of the algorithm with respect to all parameters (including the dimension),
\item we show that the complexity is polynomial-time computable in the accuracy of the output,
\item we show that the complexity is polynomial in the length\footnote{This is not
    exactly the length, more details are given in Definition \ref{def:pivp_len} and Lemma \ref{lem:ineq_I_L}} of the curve $y$ from $0$ to $t$,
\item our algorithm does not need to know this length in advance.
\end{itemize}


However the complexity of our algorithm is exponential in the dimension of the system. This is to be expected
because evaluating the truncated Taylor series is polynomial in the number of derivatives
but exponential in the dimension. Unless some breakthrough is achieved in this area,
it seems unlikely that a polynomial time algorithm in the dimension will ever be found.

Note that we are measuring computational complexity against the length of the solution curve.
We believe this parameter has a natural geometrical interpretation
and suggests that the best we can do to solve \eqref{eq:pivp} is to ``follow''
the curve, and thus the complexity of the algorithm is related to the distance we have traveled,
that is the length of the curve.

Finally, our algorithm does not need to know in advance a bound on the length of the curve:
it can automatically discover it. In this case, the complexity of the algorithm
is not known in advance but we know that the running time is polynomial in this (unknown)
quantity. Finally note that our algorithm has to assume the existence of the solution $y$
up to time $t$ because the existence is undecidable even for PIVPs (although this problem is semi-decidable) \cite{GBC07}.

\subsection{Related work}\label{sec:solve:related_work}

There are many results about the computational complexity of solving ODEs of the form:
\begin{equation}
\left\{
\begin{array}{@{}r@{}l}
y^{\prime}(t)&=f(t,y(t)) \\
y(t_0)&=y_0
\end{array}
\right.\label{eq:ivp}
\end{equation}
However, with very few exceptions,
those results assume that the ODE is solved for $t\in I=[a,b]$, i.e. a compact time domain.
This is a very convenient hypothesis for several reasons. First any open-ball convering of $I$
will admit a finite subcovering. This has been used in \cite{MM93} to show polynomial time computability in some
cases but the existence of this subcovering is not effective.
Second if $I$ is compact then $y$ is bounded and if $f$ is $C^1$ then it is always Lipschitz\footnote{
We recall that a function is Lipschitz over $I$ if $\|f(x)-f(y)\|\leq K\|x-y\|$ for some constant $K$ and $x,y\in I$}
over this bounded set and in this case the proof of the existence-uniqueness theorem for ODEs (the Picard-Lindel\"{o}f theorem) 
provides a method (Picard's iterations) to compute the solution over the compact $I$.

The reason to focus on the Lipschitz condition for $f$ seems
to be fundamental when studying the computational complexity of \eqref{eq:ivp}.
Indeed it is well-known (see e.g.~\cite[Theorem 7.3]{Ko91}) that if $f$ is not Lipschitz,
then the solution of \eqref{eq:ivp} can have arbitrarily high complexity,
even if $f$ is assumed to be polynomial-time (precision-)computable and \eqref{eq:ivp} has a unique solution.
The Lipschitz condition plays an instrumental role in the complexity because it is used to derive
the number of steps of the algorithm, for example in Picard-Lindel\"{o}f theorem
it is used to bound the number of iterations.
It was recently shown \cite{Kaw10}, following an open problem from Ko \cite[Section 7.2]{Ko91},
that if $f$ is Lipschitz and polynomial-time (precision-)computable, then the solution of \eqref{eq:ivp}
over $I=[0,1]$
can still be PSPACE-complete (but it was already known not to be of higher complexity \cite{Ko91}).


This implies that if $f$ is polynomial-time computable and Lipschitz, the complexity of computing $y$ can still be very high (PSPACE-hard) even for $I=[0,1]$. 

On the other
hand, such extreme examples are very particular and artifical, which suggests that
putting a Lipschitz condition on $f$, although convenient from a mathematical point of view,
is a very poor way of bounding the complexity for $y$.
Second, if \eqref{eq:ivp}
is solved over an unbounded time domain, $y$ may be unbounded
which in turns means that $f$ needs to be defined over an unbounded domain. This case is of
course even harder that the bounded case where the (time) domain is e.g.~$I=[0,1]$, but some results still hold. If the function $f$ in \eqref{eq:ivp} is Lipschitz over $\R$,
then the solution of \eqref{eq:ivp}
is computable over $\R$ \cite{Abe70}, \cite{Abe80}, \cite{Ko91}. Of course, requiring a global Lipschitz condition for
$f$ over the real line is a very restrictive condition. Alternatively if an effective bound
for $f$ is known or a very restricted condition on its growth is met, the solution is also computable \cite{Ruo96}.
However these cases exclude most of the interesting initial value problems. If $f$ is only locally Lipschitz then
$y$ is computable assuming it is unique \cite{CG09} but can have arbitrarily high complexity.
A workaround is to require $f$ is to be effectively locally Lipschitz but the exact complexity
in this case has not been studied as far as we are aware \cite{GZB07}.
Third, the problem of computing $I$, or even deciding of $I$ is bounded or not, is undecidable \cite{GZB07}.
This means that even if an effective Lipschitz bound is known for $f$, it is not
at all immediate how it can be used since we cannot even compute $I$ and even less a bound on $y$.

In \cite{BGP12} we have shown that a different kind of restriction on \eqref{eq:ivp}
allows for a finer complexity analysis over unbounded time domain. More precisely,
if $f$ is a polynomial (i.e. $y$ is solution of \eqref{eq:pivp})
then $y(t)$ can be computed in time polynomial in $t$,
in the accuracy of the result, and in $\sigmap{p}\max_{u\in[0,t]}\|y(u)\|^{\degp{p}}$.
More precisely, its running time is polynomial in the product $t\sigmap{p}\max_{u\in[0,t]}\|y(u)\|^{\degp{p}}$
(and the size of the coefficients).
However this result is not satisfactory for several reasons. First, and we insist on this point, it requires some
\emph{a priori knowledge of the ``bounding box'' of $y$}, i.e.~a bound on $\max_{u\in[0,t]}\|y(u)\|$, which is only semicomputable
in general. This means the algorithm in \cite{BGP12} is unusable without some knowledge of the result.
Second, this result is in some sense a worst-case scenario: if $y(t)^{\degp{p}}$ ``spikes'' 
and then becomes really small, then the resulting complexity will be high. 
Here we present an algorithm where $y(t)$ can be computed in time polynomial in $\sigmap{p}\int_{0}^t\max(1,\|y(u)\|)^{\degp{p}}du$,
which is related to the ``length'' of the solution curve (see Definition \ref{def:pivp_len}).
The following examples illustrate that the difference in complexity can be huge.

\begin{example}[Spiking function with fixed-point]
Consider the following system of differential equations, where $M\geqslant0$ is a parameter.
\begin{equation}
\left\{\begin{array}{@{}r@{}l}y(0)&=0\\z(0)&=1\end{array}\right.\qquad
\left\{\begin{array}{@{}r@{}l}y'(t)&=Mz(t)-y(t)\\z'(0)&=-z(t)\end{array}\right.
\end{equation}
It can be seen that the solution is given by $y(t)=Mte^{-t}$ which has a maximum of $Me^{-1}$
which is reached at $t=1$ and then quickly converges to $0$. Also note that the size of
the coefficients used to describe the system is of the order of $\log M$ and the degree
of the polynomial is $1$. Let us compare the asymptotic complexity of the algorithms to compute
$y(t)$ with one bit of precision. Note that both algorithms have a polynomial
dependency in $\log\sigmap{p}=\log M$.
\begin{itemize}
\item The algorithm in \cite{BGP12} has polynomial time in $t\sigmap{p}\max_{u\in[0,t]}\|y(u)\|\approx M^2t$.
\item The new algorithm has polynomial time in $\sigmap{p}\int_{0}^t\max(1,\|y(u)\|)du\approx M^2+Mt$.
\end{itemize}
It is clear that for large values of $M$ or $t$, the pratical difference in complexity
will be important because we went from a quadratic one to a linear one.
Of course this comparison should bit taken with a grain of salt:
it depends on the exact exponent of the polynomial involved\footnote{Although in this case
the two algorithms are close enough so that the comparison is relevant.}. On the other hand,
it illustrates an interesting phenomenon: the new algorithm has a linear dependency in $t$
only because of the ``$\max$'' in the integral. This means that if one could get rid
of the maximum so that the algorithm depends on the actual length of the curve, its
running time would not depend on $t$ in this example. This would make a lot of sense because
the system converges to a fixed-point so only the precision really matters for big enough $t$.
\end{example}

Finally, this problem has been widely studied in Numerical Analysis but the point of view
is usually different and more focused on practically fast algorithms rather than
asymptotically efficient algorithms. Some work \cite{IlieSC08,Corless02,Wers79,Smi1}
suggests that any polynomial time algorithm must have variable
order,
and that adaptive algorithms are theoretically superior to non-adaptive
ones\footnote{This is in contrast with classical results which state the contrary but
under usually unrealistic hypothesis, see \cite{Corless02}}. While adaptive algorithms have been widely studied, high-order
algorithms are not mainstream because they are expensive in practice. Variable order
methods have been used in \cite{CC82}, \cite{JZ05}, \cite{Smi1}, \cite{BRAB11}, \cite{ABBR12}
and some polynomial time algorithms have been obtained over compact domains or over
arbitrary domains but with stronger hypothesis.

All these results suggest that we are in fact still lacking the proper parameter against which
the complexity of $y$ should be measured.

\section{Preliminaries}\label{Sec:PreliminaryResults}

\subsection{Notations and basic facts}\label{sec:not_basics}

In this paper $\N,\Z,\Q,\R$ denote respectively the set of natural numbers, integers,
rational numbers and real numbers. We will make heavy use of the infinite norm over $\R^d$ defined
as follows:
\[\infnorm{x}=\max_{1\leqslant i\leqslant d}|x_i|\qquad x\in\R^d\]
We use a shortcut notation for a Taylor approximation:
\[\taylor{a}{n}{f}(t)=\sum_{k=0}^{n-1}\frac{f^{(k)}(a)}{k!}(t-a)^k\]
For any ring $\A$, we denote by $\A[\R^n]$ the set of multivariate polynomial
functions with $n$ variables and coefficients in $\A$. Note that $\A^d[\R^n]$
and $\A[\R^n]^d$ are isomorphic and denote the set of multivariate polynomial vectors.

We will frequently need to express complexity measures in terms of the ``size'' of
a rational number, which we define as follows:
\[\Rsize{\tfrac{p}{q}}=\Rsize{p}+\Rsize{q}\qquad\Rsize{p}=\max(1,\log p)\qquad p,q\in\Z\]

\noindent We will consider the following initial-value problem\footnote{Note that an ODE of the type $y^{\prime}=f(t,y)$
can always be reduced to an ODE $y^{\prime}=g(y)$ without (explicit) dependence on $t$ by replacing $t$
by a new variable $y_{n+1}$ defined by $y_{n+1}^{\prime}=1$.}:
\begin{equation*}\left\{\begin{array}{@{}c@{}l}y^{\prime}(t)&=p(y(t))\\y(t_0)&=y_0\end{array}\right.\end{equation*}
where $p\in\R^d[\R^n]$ is a vector of multivariate polynomials with real coefficients.
If $p\in\A[\R^n]$ is a polynomial, we write:
\[p(x_1,\ldots,x_n)=\sum_{|\alpha|\leqslant k}a_\alpha x^\alpha
\qquad \sigmap{p}=\sum_{|\alpha|\leqslant k}|a_\alpha|\]
where $k=\degp{p_i}$ is the degree of $p_i$ and $|\alpha|=\alpha_1+\cdots+\alpha_j$ as usual.
If $p\in\A^d[\R^n]$ is a vector of polynomials,
we write $\degp{p}=\max(\degp{p_1},\ldots,\degp{p_d})$ and $\sigmap{p}=\max(\sigmap{p_1},\ldots,\sigmap{p_d})$.

If $y$ is the solution of \eqref{eq:pivp}, we write $y=\Phi_p(t_0,y_0)$
so in particular $y(t)=\Phi_p(t_0,y_0)(t)$. When it is not ambiguous, we do not mention $p$.
Note that since \eqref{eq:pivp} is autonomous, $\Phi_p(t_0,y_0)(t)=\Phi_p(0,y_0)(t-t_0)$.

Finally we recall the following well-known result about arithmetico-geometric sequences.

\begin{lemma}[Arithmetico-geometric sequence]\label{lem:rec_seq_geom_arith}
Let $(a_k)_k,(b_k)_k\in\R^\N$ and assume that $u\in\R^\N$ satisfies $u_{n+1}= a_nu_n+b_n$ for all $n\in\N$.
Then:
\[u_n=u_0\prod_{i=0}^{n-1}a_i+\sum_{i=0}^{n-1}b_i\prod_{j=i+1}^{n-1}a_j\]
\end{lemma}

\subsection{Complexity Model}\label{sec:complexity_model}

In this work, $\mathcal{O}$ denotes the usual ``big-O'' relation and $\tilde{\mathcal{O}}$ the
``soft-O'' relation. In this paper, we use three different notions of complexity which we describe below.

The \emph{algebraic (or arithmetic) complexity} (denoted as $C_{arith}$)
is the number of basic operations (addition, multiplication, division and comparison) in the base field $\R$,
using rational constants.

The \emph{rational bit complexity} (denoted as $C_{\Q,bit}$) is the classical complexity of Turing machines
with rational inputs\footnote{Recall that rational inputs
can be encoded as a pair of integers for example and the exact encoding does influence
the size of the encoding up to a polynomial.}. The size of a rational number $x$
will be denoted by $\Rsize{x}$, see previous section for the exact definition.


\subsection{Dependency in the parameters}\label{sec:ode_dependency}

We recall some useful results about the Lipschitz bounds for polynomials
and the dependency of PIVP in the initial value.

\begin{lemma}[Effective Lipschitz bound for polynomials \cite{BGP12}]\label{lem:multivariate_poly_lipschitz}
Let $p\in\R^n[\R^d]$ and $k=\degp{p}$ its degree. For all $a,b\in\R^d$ such that $\infnorm{a},\infnorm{b}\leqslant M$,
$\infnorm{p(b)-p(a)}\leqslant kM^{k-1}\sigmap{p}\infnorm{b-a}$.
\end{lemma}

\begin{theorem}[Initial value dependency of PIVP]\label{th:dependency_right_side}
Let $I$ be an interval, $p\in\R^n[\R^{n}]$, $k=\degp{p}$ and $y_0,z_0\in\R^d$.
Assume that $y,z:I\rightarrow\R^d$ satisfy:
\[\left\{\begin{array}{@{}r@{}l}y(a)&=y_0\\y'(t)&=p(y(t))\end{array}\right.\qquad
\left\{\begin{array}{@{}r@{}l}z(a)&=z_0\\z'(t)&=p(z(t))\end{array}\right.\qquad t\in I\]
Define for all $t\in I$:
\[
\mu(t)=\infnorm{z_0-y_0}
\exp\left(k\sigmap{p}\int_{a}^t(\varepsilon+\infnorm{y(u)})^{k-1}du\right)
\]
If $\mu(t)<\varepsilon$ for all $t\in I$, then $\infnorm{z(t)-y(t)}\leqslant\mu(t)$ for all $t\in I$.
\end{theorem}

\begin{proof}Similar to the proof of Proposition 3 in \cite{BGP12}.
\end{proof}

\subsection{Taylor series of the solutions}\label{sec:ode_taylor}

It is well-known \cite[Section 32.4]{Arn78} that solutions of a PIVP
are analytic so in particular the Taylor series at any point converges.
This yields the natural question of the rate of convergence
of the series, and the complexity of computing the truncated series.

In the case of a function satisfying a polynomial differential equation like
\eqref{eq:pivp}, we can obtain a sharper bound than the one given by the classical
Taylor-Lagrange theorem. These bounds are based on Cauchy majorants of series and
we refer the reader to \cite{BGP12} and \cite{WWSPC06} for the details.

\begin{theorem}[Taylor approximation for PIVP]\label{th:wws06}
If $y$ satisfies \eqref{eq:pivp} for $t_0=0$, $k=\degp{p}\geqslant2$, $\alpha=\max(1,\infnorm{y_0})$,
$M=(k-1)\sigmap{p}\alpha^{k-1}$, $|t|<\frac{1}{M}$ then
\[\infnorm{y(t)-\taylor{0}{n}{y}(t)}\leqslant\frac{\alpha|Mt|^n}{1-|Mt|}\]
\end{theorem}

The next problem we face is to compute the truncated Taylor series of
the solution over a small time interval. In this paper,
we will assume that we have access to a subroutine $\operatorname{ComputeTaylor}$
which computes this truncated Taylor series (see Algorithm \ref{alg:taylor}).

\begin{algorithm}
\begin{algorithmic}[1]
\Require{$p\in\R^d[\R^d]$ the polynomial of the PIVP}
\Require{$y_0\in\R^d$ the initial condition}
\Require{$\omega\in\N$ the order of the approximation}
\Require{$\varepsilon\in]0,1]$ the accuracy requested}
\Require{$t\in\R$ the time step}
\Function{ComputeTaylor}{$p,y_0,\omega,\varepsilon,t$}
\State\Return{x}\Comment{such that $\infnorm{x-\taylor{0}{\omega}y(t)}\leqslant\varepsilon$ where $y(0)=y_0$ and $y'=p(y)$}
\EndFunction
\end{algorithmic}
\caption{Taylor Series algorithm for PIVP\label{alg:taylor}}
\end{algorithm}

The complexity of computing this Taylor series has already been analyzed in the literature.
Let $\operatorname{TL}(d,p,y_0,\omega,\varepsilon,t)$ be the complexity of \algoref{alg:taylor}.
More precisely, we will refer to the \emph{bit-complexity} as $\operatorname{TL}_{bit}$
and the \emph{arithmetic complexity} as $\operatorname{TL}_{arith}$.

\begin{theorem}The complexity of \algoref{alg:taylor} is bounded by:
\begin{equation}\label{eq:tl_arith}
\operatorname{TL}_{arith}=\poly(\omega,d,k^d)
\end{equation}
\begin{equation}\label{eq:tl_bit}
\operatorname{TL}_{\Q,bit}=\poly((k\omega)^d,\Rsize{t},\Rsize{\sigmap{p}},\Rsize{\infnorm{y_0}},-\log\varepsilon)
\end{equation}
\end{theorem}

\begin{proof}
The first result uses Theorem 3 of \cite{BostanChyzakOllivierSalvySchostSedoglavic2007}
which shows that the complexity of computing the first $\omega$ terms of the series of the solution $y$ of
\[y'=\varphi(t,\phi)\qquad y(0)=v\]
where $\phi$ is a vector of dimension $d$ of multivariate power series with coefficients in $\R$,
is bounded by
\begin{equation}\label{eq:generic_series_comp_bound}\bigO{L(\omega)+\min(\texttt{MM}(d,\omega),d^2\texttt{M}(\omega)\log \omega)}\end{equation}
where $\texttt{MM}(d,k)$ is the arithmetic cost of multiplying two $d\times d$ matrices with
polynomial entries with degree less than $k$, $\texttt{M}(k)=\texttt{MM}(1,k)$
and $L(\omega)$ is the cost of computing the first $\omega$ terms of the composition $\varphi(t,s(t))$
and $\mathbf{Jac}(\varphi)(t,s(t))$ for any powers series $s$.

In our case, $\phi(t,s(t))=p(s(t))$ is a polynomial so computing the first $\omega$ terms of $p(s(t))$
and $\mathbf{Jac}(p)(s(t))$ costs a number of arithmetical operations polynomial in $\omega$
and $k^d$ (the maximum number of coefficients in $p$) where $k$ is the degree of $p$. In other words:
\[L(\omega)=\poly(\omega,k^d)\]
Furthermore, \cite{BostanChyzakOllivierSalvySchostSedoglavic2007} mentions that we can always
choose $\texttt{MM}(d,k)$ in $\bigO{d^a\texttt{M}(k)}$
where $a$ is the matrix multiplication exponent, and $\texttt{M}(k)=\bigO{k\log k\log\log k}$. In
other words:
\[\texttt{M}(k)=\poly(k)\qquad\texttt{MM}(d,k)=\poly(d,k)\]
Putting everything together we get that \eqref{eq:generic_series_comp_bound} is bounded
by a polynomial in $\omega$, $d$ and $k^d$. Once we have the first $\omega$ terms,
we can easily evaluate the truncated Taylor series with a number of operations polynomial
in $\omega$. In other words:
\[\operatorname{TL}_{arith}=\poly(\omega,d,k^d)\]

%

In \cite{BGP12} we described a very naive way of implementing this algorithm,
showing that the rational bit complexity is bounded by:
\[
\operatorname{TL}_{\Q,bit}=\poly((k\omega)^d,\log\max(1,t)\sigmap{p}\max(1,\infnorm{y_0}),-\log\varepsilon)
\]
More explicit formulas can be found in \cite{MM93}.


\end{proof}

Notice that there is a significant difference
between the two complexity notions: the bit-complexity depends on the accuracy $\varepsilon$ whereas the arithmetic
complexity does not. Indeeed, in the arithmetic model the computation is exact whereas
in the Turing model only finite approximations can be computed in finite time. Also note
that the bit-complexity result gives an implicit bound on the size of result, which is
an interesting fact by itself.

As a side note, it is worth noting that for purely rational inputs,
it would be possible to compute the truncated Taylor series with infinite precision
($\varepsilon=0$) because the result is a rational number \cite{BGP12}. However this is not helpful
in general because without a precision control of the output precision, no algorithm would have polynomial
running time. For example, if the exact output precision is double the one of the input,
after $k$ iterations, the number of digits of the output is exponential in $k$.

\section{The generic Taylor method}\label{sec:generic_taylor}

We consider a generic adaptive Taylor meta-algorithm to numerically solve \eqref{eq:pivp}.
This is a meta-algorithm in the sense that we do not specify, for now, how the parameters are chosen.
As a matter of fact, most Taylor algorithms can be written this way, including the Euler method.
The goal of this algorithm is to compute $x\in\R^d$ such
that $\infnorm{x-y(t)}<\varepsilon$ where $y$ satisfies \eqref{eq:pivp}, given as inputs $t\in\R$, $\varepsilon\in]0,1]$,
$p$ and $y_0$.
We assume that the meta-algorithm uses the following parameters:
\begin{itemize}
\item $n\in\N$ is the number of steps of the algorithm
\item $t_0<t_1<\ldots<t_n=t$ are the intermediate times
\item $\delta t_i=t_{i+1}-t_i$ are the time steps
\item for $i\in\intinterv{0}{n-1}$, $\omega_i\in\N$ is the order at time $t_i$ and $\mu_i>0$ is the rounding error at time $t_i$
(see \eqref{eq:alg_step})
\item $\tilde{y}_i\in\R^d$ is the approximation of $y$ at time $t_i$
\item $\varepsilon_i\geqslant\infnorm{y(t_i)-\tilde{y}_i}$ is a bound on the error at step $i$
\end{itemize}

This meta-algorithm works by solving the ODE \eqref{eq:pivp} with initial condition
$y(t_i)=\tilde{y}_i$ over a small time interval $[t_i,t_{i+1}]$,
yielding as a result the approximation $\tilde{y}_{i+1}$ of $y(t_{i+1})$.
This approximation over this small time interval is obtained using the algorithm
of \secref{sec:ode_taylor}, through a Taylor approximation of order $\omega_i$. This procedure is repeated over 
$[t_0,t_1], [t_1,t_2],\dots,[t_i,t_{i+1}],\dots$ until we reach the desired time $t_n=t$. 
Therefore the meta-algorithm is only assumed to satisfy the following inequality at each step:
\begin{equation}\label{eq:alg_step}
\infnorm{\tilde{y}_{i+1}-\taylor{t_i}{\omega_i}\Phi(t_i,\tilde{y}_i)(t_{i+1})}\leqslant\mu_i
\end{equation}
Note in particular that this implies that the solution at $t_{i+1}$ can be computed via the Taylor series $\taylor{t_i}{\omega_i}\Phi(t_i,\tilde{y}_i)$.
A direct consequence of \eqref{eq:alg_step} and the triangle inequality is that:
\begin{equation}\label{eq:eps_ineq_1}
\begin{aligned}
\varepsilon_{i+1}
&\leqslant\infnorm{y(t_{i+1})-\Phi(t_i,\tilde{y}_i)(t_{i+1})}\\
&+\infnorm{\Phi(t_i,\tilde{y}_i,t_{i+1})-\taylor{t_i}{\omega_i}\Phi(t_i,\tilde{y}_i)(t_{i+1})}\\
&+\mu_i
\end{aligned}
\end{equation}

The first term on the right-hand side is the the global error: it arises because after one step,
the solution we are computing lies on a different solution curve than the true solution.
The second term is the local (truncation) error: at each step we only compute a truncated Taylor series
instead of the full series. The third error is the rounding error: even if we truncate the
Taylor series and evaluate it, we only have a finite number of bits to store it
in the Turing model and this rounding introduces an error.

In order to bound the first two quantities, we will rely on the results of the
previous sections. Since those results only hold for reasonable (not too big) time steps, we
need to assume bounds on the time steps. To this end, we introduce the following quantities:
\begin{equation}\label{eq:beta_gamma}
\begin{aligned}
\beta_i&=k\sigmap{p}\max(1,\infnorm{\tilde{y}_i})^{k-1}\delta t_i\\
\gamma_i&=\int_{t_i}^{t_{i+1}}k\sigmap{p}(\varepsilon+\infnorm{y(u)})^{k-1}du
\end{aligned}
\end{equation}
where $\delta t_i=t_{i+1}-t_i$. The choice of the values for $\beta_i$ and $\gamma_i$
comes from \thref{th:wws06} and \thref{th:dependency_right_side}, respectively.
We assume that the meta-algorithm satisfies the following extra property:
\begin{equation}\label{eq:assumption_1}
\beta_i<1
\end{equation}
Back to \eqref{eq:eps_ineq_1}, we apply \thref{th:wws06} and \thref{th:dependency_right_side} to get
\begin{equation}\label{eq:eps_ineq_2}
\varepsilon_{i+1}\leqslant\varepsilon_i e^{\gamma_i}+\frac{\max(1,\infnorm{\tilde{y}_i})\beta_i^{\omega_i}}{1-\beta_i}+\mu_i\\
\end{equation}

We can now apply \lemref{lem:rec_seq_geom_arith} to \eqref{eq:eps_ineq_2},
since all quantities are positive, therefore obtaining a bound on $\varepsilon_n$.
We further bound it using the fact that $\prod_{j=i+1}^{n-1}a_j\leqslant\prod_{j=0}^{n-1}a_j$ 
when $a_j\geqslant1$. This gives a bound on the error done by the generic Taylor algorithm:
\begin{equation}\label{eq:eps_ineq_3}
\begin{aligned}
\varepsilon_n&\leqslant(\varepsilon_0 +B)e^A\\
A&=\sum_{i=0}^{n-1}\gamma_i=\int_{t_0}^{t_n}k\sigmap{p}(\varepsilon+\infnorm{y(u)})^{k-1}du\\
B&=\sum_{i=0}^{n-1}\frac{\max(1,\infnorm{\tilde{y}_i})\beta_i^{\omega_i}}{1-\beta_i}+\sum_{i=0}^{n-1}\mu_i
\end{aligned}
\end{equation}

In other words, assuming that $\beta_i<1$ yields a generic error bound on the output of the algorithm.
This leaves us with a large space to play with and optimize the parameters ($\beta_i$, $\gamma_i$, $\mu_i$)
to get a correct and efficient algorithm.

\section{The adaptive Taylor algorithm}\label{Sec:Algorithm}

In this section, we instantiate the generic algorithm of the previous section
using carefully chosen parameters to optimize its complexity. In order to analyze the algorithm, it is
useful to introduce the following quantity:

\begin{equation}\label{eq:hyp_I_lambda}
\IntI(t_0,t)=\int_{t_0}^{t}k\sigmap{p}\max(1,\varepsilon+\infnorm{y(u)})^{k-1}du
\end{equation}

This algorithm will be parametrized by the choice of a ``hint'' $I$
and the number of steps $n$.

Equation \eqref{eq:eps_ineq_3}
suggests that it is advantageous to choose $\beta_i$ smaller than $1$ so that
$\beta_i^{\omega_i}$ is small and not too small so that the number of steps doesn't blow up due to \eqref{eq:beta_gamma}. That's why we introduced the condition \eqref{eq:assumption_1}.
Since a reasonable choice for this parameter is nontrivial, we introduce another
parameter $\lambda$ to be fixed later and we assume that $\beta_i=\lambda$, except
for the last time step which might require a smaller value to stumble exactly on the final time. We assume that:

\begin{equation}\label{eq:parameters}
n>0\qquad I>0\qquad 0\leqslant\lambda\leqslant\frac{1}{2}
\end{equation}

Let us now establish the other parameters of the algorithm. We define the following values:

\begin{equation}\label{eq:choice_deltat}
\delta t_i=\min\left(t-t_{i},\frac{\lambda}{k\sigmap{p}\max(1,\infnorm{\tilde{y}_i})^{k-1}}\right)
\end{equation}
\begin{equation}\label{eq:choice_omega_mu_eps_eta}
\omega_i=\log_2\frac{6n\max(1,\infnorm{\tilde{y}_i})}{\eta}
\qquad
\mu_i=\frac{\eta}{3n}
\qquad
\varepsilon_0\leqslant\frac{\varepsilon}{3}e^{-I}
\qquad
\eta=\varepsilon e^{-I}
\end{equation}

We will now see that under natural constraints on $I$, the algorithm will be correct.

\begin{lemma}[Algorithm is conditionally correct]\label{lem:algo_correct}
If $I\geqslant\IntI(t_0,t_n)$ then the choices of the parameters, implied by \eqref{eq:choice_deltat} and
\eqref{eq:choice_omega_mu_eps_eta} guarantee that $\varepsilon_n\leqslant\varepsilon$.
\end{lemma}

\begin{proof}
We only need to put all pieces together using \eqref{eq:eps_ineq_3}. Indeed, it is clear that
\begin{equation}
A\leqslant\IntI(t_0,t_n)\leqslant I
\end{equation}
Furthermore the way how we chose $\lambda$ implies $\beta_i\leqslant\lambda$,
and together with \eqref{eq:choice_omega_mu_eps_eta} and \eqref{eq:hyp_I_lambda}, it implies that:
\begin{equation*}
\frac{\beta_i^{\omega_i}}{1-\beta_i}\leqslant\frac{\lambda^{\omega_i}}{1-\lambda}
\leqslant \frac{2^{-\omega_i}}{1/2}\leqslant \frac{2\eta}{6n\max(1,\infnorm{\tilde{y}_i})}
\end{equation*}
Thus:
\begin{align*}
\varepsilon_n&\leqslant(\varepsilon_0 +B)e^A\\
&\leqslant \frac{\varepsilon}{3}e^{-I}e^{I}+e^I\left(\sum_{i=0}^{n-1}\frac{\eta}{3n}+\sum_{i=0}^{n-1}\frac{\eta}{3n}\right)\\
&\leqslant \frac{\varepsilon}{3}+e^I\frac{2\eta}{3}\leqslant \varepsilon
\end{align*}
\end{proof}

This result shows that if we know $\IntI(t_0,t)$ then we can compute the solution
numerically by plugging $I=\IntI(t_0,t)$ in the algorithm.
However, this result does not tell us anything about how to find it.
Furthermore, a desirable property of the algorithm would be to detect if the argument
$I$ is not large enough, instead of simply returning garbage. Another point
of interest is the choice of $n$: if we pick $n$ too small in the algorithm,
it will be correct but $t_n<t$, in other words we won't reach the target time.

Both issues boil down to the lack of a relationship between $I$ and $n$: this is the point of
the following lemma. We will see that it also explains how to choose $\lambda$.

\begin{lemma}[Relationship between $n$ and $I$]\label{lem:bound_n_I}
Assume that $k\geqslant2$ and choose some $\lambda$, $I$, and $\varepsilon$ satisfying 
\begin{equation*}
\frac{1}{\lambda}\geqslant1+\frac{k}{1-2k\varepsilon}
\qquad
I\geqslant\IntI(t_0,t_n)
\qquad
\varepsilon\leqslant\frac{1}{4k}
\end{equation*}
Then for all $i\in\intinterv{0}{n-1}$,
\begin{equation*}
\beta_i(1-e^{-1})\leqslant \IntI(t_i,t_{i+1})\leqslant \beta_i e,
\end{equation*}
so in particular:
\begin{equation*}
\tfrac{1}{2}\beta_i\leqslant \IntI(t_i,t_{i+1})\leqslant 3\beta_i.
\end{equation*}
\end{lemma}

\begin{proof}
Note that the hypothesis on $\varepsilon$ is mostly to make sure that fraction on the right-hand side of the condition for $\lambda$ is well defined.
Pick any $u\in\left[t_i,t_{i+1}\right]$ and
 apply \thref{th:dependency_right_side} with 
$\mu(u)\leqslant \varepsilon_{i}e^{\gamma_{i}}\leqslant\varepsilon$
to get that:
\begin{equation*}
\infnorm{y(u)-\Phi_p(t_i,\tilde{y}_{i})(u)}\leqslant\varepsilon
\end{equation*}
Furthermore, for any such $u$, apply \thref{th:wws06} with $n=1$ to get:
\begin{equation*}
\infnorm{\tilde{y}_{i+1}-\Phi_p(t_i,\tilde{y}_{i})(u)}
\leqslant\frac{\alpha|M(u-t_{i})|}{1-|M(u-t_{i})|}
\end{equation*}
where
$M=k\sigmap{p}\alpha^{k-1}$ and $\alpha=\max(1,\infnorm{\tilde{y}_{i}})$.
Putting everything together we get for $\xi=\frac{u-t_i}{\delta t_i}$ using \eqref{eq:beta_gamma}:
\begin{equation*}
\infnorm{y(u)-\tilde{y}_i}\leqslant\varepsilon+\frac{\alpha M(u-t_i)}{1-M(u-t_i)}
\leqslant\varepsilon+\frac{\alpha \beta_i\xi}{1-\beta_i\xi}\leqslant
\varepsilon+\frac{\alpha \lambda\xi}{1-\lambda\xi}
\end{equation*}
Consequently:
\begin{equation*}
\infnorm{\tilde{y}_i}-\frac{\alpha \lambda\xi}{1-\lambda\xi}
\leqslant\varepsilon+\infnorm{y(u)}
\leqslant2\varepsilon+\infnorm{\tilde{y}_i}+\frac{\alpha \lambda\xi}{1-\lambda\xi}
\end{equation*}
And since $\alpha=\max(1,\infnorm{\tilde{y}_i})$:
\begin{equation*}
\infnorm{\tilde{y}_i}-\frac{\alpha \lambda\xi}{1-\lambda\xi}
\leqslant\varepsilon+\infnorm{y(u)}
\leqslant2\varepsilon+\alpha\left(1+\frac{\lambda\xi}{1-\lambda\xi}\right)
\end{equation*}
And a case analysis brings:
\begin{equation*}
\max\left(1,\alpha-\frac{\alpha\lambda\xi}{1-\lambda\xi}\right)
\leqslant\max(1,\varepsilon+\infnorm{y(u)})
\leqslant\max\left(1,2\varepsilon+\alpha+\frac{\alpha\lambda\xi}{1-\lambda\xi}\right)
\end{equation*}
Which can be overapproximated by:
\begin{equation*}
\alpha-\frac{\alpha\lambda\xi}{1-\lambda\xi}
\leqslant\max(1,\varepsilon+\infnorm{y(u)})
\leqslant2\varepsilon+\alpha+\frac{\alpha\lambda\xi}{1-\lambda\xi}
\end{equation*}
And finally:
\begin{equation*}
\left(\alpha-\frac{\alpha\lambda\xi}{1-\lambda}\right)^{k-1}
\leqslant\max(1,\varepsilon+\infnorm{y(u)})^{k-1}
\leqslant\left(2\varepsilon+\alpha+\frac{\alpha\lambda\xi}{1-\lambda}\right)^{k-1}
\end{equation*}
A simple calculation shows that $\int_0^1(a+bu)^{k-1}du=\frac{(a+b)^k-a^k}{bk}$.
Integrating the previous bounds over $[t_i,t_{i+1}]$, we get, for $x=\frac{\lambda}{1-\lambda}$:
\begin{equation*}
\int_{t_i}^{t_{i+1}}\left(\alpha-\frac{\alpha\lambda\xi}{1-\lambda}\right)^{k-1}du
=\alpha^{k-1}\delta t_i\frac{1-(1-x)^k}{kx}
\end{equation*}
Realising that $x=\frac{1}{\frac{1}{\lambda}-1}$, the hypothesis on $\lambda$ yields:
\begin{equation}\label{eq:ineq_x_k}
x\leqslant\frac{1-2k\varepsilon}{k}\leqslant\frac{1}{k}
\end{equation}
A simple analysis of the function $x\mapsto \frac{1-(1-x)^k}{kx}$ shows that it is
decreasing on $]0,\frac{1}{k}]$ and so satisfies, for $x$ in this interval,
\begin{equation*}
\frac{1-(1-x)^k}{kx}\geqslant1-\left(1-\frac{1}{k}\right)^k\geqslant1-e^{-k\ln(1-\frac{1}{k})}\geqslant1-e^{-1}
\end{equation*}
So finally we get:
\begin{equation}\label{eq:main_ineq_int_left}
\int_{t_i}^{t_{i+1}}\left(\alpha-\frac{\alpha\lambda\xi}{1-\lambda}\right)^{k-1}du
\geqslant \alpha^{k-1}\delta t_i(1-e^{-1})
\end{equation}
On the other side, we get:
\begin{equation*}
\int_{t_i}^{t_{i+1}}\left(2\varepsilon+\alpha+\frac{\alpha\lambda\xi}{1-\lambda}\right)^{k-1}du
=\alpha^{k-1}\delta t_i\frac{(2\frac{\varepsilon}{\alpha}+1+x)^k-(2\frac{\varepsilon}{\alpha}+1)^k}{kx}
\end{equation*}
And since $\alpha\geqslant1$ and $b^k-a^k\leqslant k(b-a)b^{k-1}$ when $b\geqslant a$, we get:
\begin{equation*}
\int_{t_i}^{t_{i+1}}\left(2\varepsilon+\alpha+\frac{\alpha\lambda\xi}{1-\lambda}\right)^{k-1}du
\leqslant\alpha^{k-1}\delta t_i(2\varepsilon+1+x)^{k-1}
\end{equation*}
We can now use \eqref{eq:ineq_x_k} to get:
\begin{equation}\label{eq:main_ineq_int_right}
\int_{t_i}^{t_{i+1}}\left(2\varepsilon+\alpha+\frac{\alpha\lambda\xi}{1-\lambda}\right)^{k-1}du
\leqslant\alpha^{k-1}\delta t_i\left(1+\frac{1}{k}\right)^{k-1}
\leqslant\alpha^{k-1}\delta t_i e
\end{equation}
We can now put together \eqref{eq:main_ineq_int_left} and \eqref{eq:main_ineq_int_right} using that $M=k\sigmap{p}\alpha^{k-1}$:
\begin{equation*}
\delta t_i M(1-e^{-1})
\leqslant\int_{t_i}^{t_{i+1}}\hspace{-1em}k\sigmap{p}\max(1,\varepsilon+\infnorm{y(u)})^{k-1}du
\leqslant\delta t_i M e
\end{equation*}
which shows the result since $\beta_i=M\delta t_i$. The last inequalities trivially
follow from bounds on $e$.
\end{proof}

Next we implement the generic meta-algorithm of Section \ref{sec:generic_taylor}
to the choice of parameters specified by \ref{eq:choice_deltat} and \ref{eq:choice_omega_mu_eps_eta}
and obtain \algoref{alg:solve_pivp_taylor}.

\begin{lemma}[Algorithm is correct]\label{lem:solve_pivp_taylor_correct}
Let $t\in\R$, $I>0$ and $\varepsilon>0$, and assume that $y$ satisfies \eqref{eq:pivp} over $[t_0,t]$.
Let $x=\operatorname{SolvePIVPVariable}(t_0,y_0,p,t,\varepsilon,I)$, where $\operatorname{SolvePIVPVariable}$ is given by Algorithm \ref{alg:solve_pivp_taylor}. Then
\begin{itemize}
\item Either $x=\bot$ or $\infnorm{x-y(t)}\leqslant\varepsilon$
\item Furthermore, if $\displaystyle I\geqslant6\IntI(t_0,t)$
then $x\neq\bot$
\end{itemize}
\end{lemma}

\begin{algorithm}
\begin{algorithmic}[1]
\Require{$t_0\in\R$ the initial time}
\Require{$y_0\in\R^d$ the initial condition}
\Require{$p\in\R^d[\R^d]$ polynomial of the PIVP}
\Require{$t\in\R$ the time step}
\Require{$\varepsilon\in\R$ the accuracy requested}
\Require{$I\in\R$ the integral hint}
\Function{SolvePIVPVariable}{$t_0,y_0,p,t,\varepsilon,I$}
\State $k\algassign\max(2,\degp{p})$
\State $\varepsilon\algassign\min\left(\varepsilon,\frac{1}{4k}\right)$
\State $u\algassign t_0$
\State $\tilde{y}\algassign y_0$
\State $i\algassign0$
\State $\lambda\algassign1-\frac{k}{1-2k\varepsilon+k}$
\State $N\algassign 1+\frac{2I}{\lambda}$
\State $\eta\algassign \varepsilon e^{-I}$
\State $\mu\algassign\dfrac{\eta}{3N}$
\State $\beta\algassign0$
\While{$u<t$}\label{alg:pivp:test1}\Comment{See comments}
    \If{$i\geqslant N$}
    \State\Return{$\bot$}\Comment{Too many steps !}
    \EndIf
    \State $\delta\algassign\min\left(t-u,\dfrac{\lambda}{k\sigmap{p}\max(1,\infnorm{\tilde{y}})^{k-1}}\right)$
    \State $\beta\algassign k\sigmap{p}\max(1,\infnorm{\tilde{y}})^{k-1}\delta $
    \State $\omega\algassign\left\lceil-\log_2\dfrac{\eta}{6N\max(1,\infnorm{\tilde{y}})}\right\rceil$
    \label{alg:pivp:log}\Comment{See comments}
    \State $\tilde{y}\algassign\operatorname{ComputeTaylor}(p,\tilde{y},\omega,\mu,\delta)$
    \State $u\algassign u+\delta$
    \State $i\algassign i+1$
    \If{$I<3((i-1)\lambda+\beta)$}\label{alg:pivp:test2}\Comment{See comments}
        \State\Return{$\bot$}\Comment{Unsafe result !}
    \EndIf
\EndWhile
\State \Return{$\tilde{y}$}
\EndFunction
\Comments{For rational bit complexity}
\CommentLine{\ref{alg:pivp:test1}}{This a comparison between two rational numbers}
\CommentLine{\ref{alg:pivp:log}}{Content of the $\log_2$ is rational, so ceil is computable (see proof)}
\CommentLine{\ref{alg:pivp:test2}}{This a comparison between two rational numbers}
\end{algorithmic}
\caption{PIVP Solving algorithm: $\operatorname{SolvePIVPVariable}$\label{alg:solve_pivp_taylor}}
\end{algorithm}

\begin{proof}
Let $n$ be the number of steps of the algorithm.
It is clear that the algorithm performs exactly as described in the previous section,
in the sense that it either returns $\bot$ or $\tilde{y}_n$ with $n\leqslant N$
which satisfies $t_n=t$ and $N=1+\frac{2I}{\lambda}$. Now consider the two possible cases.

If $I\geqslant\IntI(t_0,t)$ then by \lemref{lem:algo_correct}, we get that
\begin{equation*}
\infnorm{y(t)-\tilde{y}_n}\leqslant\varepsilon
\end{equation*}
so the algorithm is correct if it returns a value instead of $\bot$. Furthermore, by \lemref{lem:bound_n_I}:
\begin{equation*}
\tfrac{1}{2}\sum_{i=0}^{n-1}\beta_i\leqslant\IntI(t_0,t_n)\leqslant 3\sum_{i=0}^{n-1}\beta_i
\end{equation*}
And since $\beta_{n-1}\leqslant\lambda$ and $\beta_i=\lambda$ for $i<n-1$ then
\begin{equation*}
\tfrac{1}{2}((n-1)\lambda+\beta_{n-1})\leqslant\IntI(t_0,t_n)\leqslant 3n\lambda
\end{equation*}
In the case of $I\geqslant6\IntI(t_0,t)$, we further have:
\begin{equation}\label{eq:n_I}
I\geqslant6\IntI(t_0,t)\geqslant 3((n-1)\lambda+\beta_{n-1})
\end{equation}
Consequently, the final test of the algoritm will fail because $\beta$ in the algorithm is exactly
$\beta_{n-1}$. Note however that the algorithm could still return $\bot$ if the
test ``$i>N$'' in the algorithm succeeds. Suppose by contradiction that this is case.
Then $t_n<t$ and $n>N$ otherwise the algorithm would have returned a value
(because $n$ is the number of steps, so it is also the final value of $i$).
In particular this implies that $\beta_{n-1}=\lambda$ so \eqref{eq:n_I} becomes:
\begin{equation*}
I\geqslant3n\lambda
\end{equation*}
But recall that $N=1+\frac{2I}{\lambda}$ so in particular we have:
\begin{equation*}
N\geqslant1+\frac{6n\lambda}{\lambda}\geqslant1+n
\end{equation*}
which is absurd because $n>N$. This means that $t_n=t$ and so the algorithm
returns $\tilde{y}_n$ which is correct, as we so above.

Now comes the case of $I<\IntI(t_0,t)$. This case is more subtle because we cannot directly
apply \lemref{lem:bound_n_I}, indeed we miss the hypothesis on $I$.
First note that we can ignore the case where $t_N<t$ because the algorithm returns $\bot$ in this case anyway.
Since the function $u\mapsto\IntI(t_0,u)$ is continuous on $[t_0,t]$ and
is $0$ on $t_0$ and $>I$ on $t$, we can apply the intermediate value theorem to get:
\begin{equation*}
\exists u\in[t_0,t[\text{ such that }I=\IntI(t_0,u)
\end{equation*}
Since we eliminated the case where $t_N<t$, we know that $t_n=t$ for some $n\leqslant N$ in the algorithm, so necessarily:
\begin{equation*}
\exists i_0\in\intinterv{0}{n-1}\text{ such that }t_{i_0}\leqslant u<t_{i_0+1}
\end{equation*}
As a consequence of $u\mapsto\IntI(t_0,u)$ being an increasing function,
\begin{equation*}
I\geqslant\IntI(t_0,t_{i_0})
\end{equation*}
Imagine for a moment that we run the algorithm again with final time $u$ instead of $t$. A close look at the code
shows that all variables (which we call $t_i'$, $\beta_i'$, and so on) will be the same for $i\leqslant i_0$ but then
$t_{i_0+1}'=u$. So we can apply \lemref{lem:bound_n_I} to this new run of the algorithm on $[t_0,t_{i_0+1}']$ to get that
\begin{equation*}
\tfrac{1}{2}\sum_{i=0}^{i_0}\beta_i'\leqslant\IntI(t_0',t_{i_0+1}')\leqslant 3\sum_{i=0}^{i_0}\beta_i'
\end{equation*}
And since $i_0<n$ then $\beta_i'=\lambda$ for $i<i_0$, and $\beta_{i_0}'<\beta_{i_0}$ because $u=t_{i_0+1}'<t$ so the equation becomes:
\begin{equation*}
\tfrac{1}{2}i_0\lambda\leqslant\IntI(t_0,t_{i_0}')\leqslant 3(i_0\lambda+\beta_{i_0}')<3(i_0\lambda+\beta_{i_0})
\end{equation*}
Since $I<\IntI(t_0,t_{i_0})$ and $t_{i_0}=t_{i_0}'$, this simplifies to
\begin{equation*}
I<3(i_0\lambda+\beta_{i_0})
\end{equation*}
Which leads to
\begin{equation*}
I<3((n-1)\lambda+\beta_{n-1})
\end{equation*}
because either $i_0=n-1$ and this trivial, or $i_0<n-1$ and then use $\beta_{i_0}<\lambda$.

Notice that this result is completely independent of the run the algorithm, we just used a ``virtual''
run of the algorithm to obtain it. Consequently, in the original algorithm, the final test
will suceed and the algorithm will return $\bot$.
\end{proof}

As we see from this lemma $I$ just needs to be big enough.
Otherwise the algorithm will either return a correct value or an error $\bot$.
One can reformulate \lemref{lem:solve_pivp_taylor_correct} as:
\begin{itemize}
\item Whatever the inputs are, we have a bound on the number of steps executed by the algorithm
\item If $I$ is greater than a specified value, we know that the algorithm will return a result (and not an error, i.e $\bot$)
\item If the algorithm returns a result $x$, then this value is correct, that is $\infnorm{x-y(t)}\leqslant\varepsilon$.
\end{itemize}

Notice that the number of steps $n$ of \lemref{lem:solve_pivp_taylor_correct} is only the number of time steps
$[t_0,t_1],[t_1,t_2],\ldots,[t_{n-1},t_n]$ used by our method. But inside each subinterval $[t_i,t_{i-1}]$ we still
have to compute the solution of \eqref{eq:pivp} over this subinterval using the Taylor approximation described in \secref{sec:ode_taylor}.
We recall that whether we use bit complexity or algebraic complexity, the complexity of
finding this Taylor approximation is polynomial in the order $\omega$ of the method, in the initial condition $y_0$, in
the accuracy $k$ (actually the accuracy is $\mu=2^{-k}$) and in the description of the polynomial $p$. Using this
result and the previous theorem, we obtain the \lemref{lem:solve_pip_var_com} about the computational complexity of solving
\eqref{eq:pivp} over unbounded domains.

In order to bound the complexity, we need to introduce another quantity which is related to $\IntI$ but actually
closer to what we are really interested in: the length of the curve $y$. We recall that the length of the curve
defined by the graph of a function $f$ between $x=a$ and $x=b$ is:
\[
 \operatorname{length} = \int_a^b\sqrt{1+\left( f^{\prime}(x)\right)^2}dx
\]

In the case of the solution of \eqref{eq:pivp}, we note that the derivative of the solution $y$ is given by $p(y)$. Since
the degree of $p$ is $k$, the length of the solution has a value which has an order of magnitude similar to
the following quantity.

\begin{definition}[Pseudo-length of a PIVP]\label{def:pivp_len}
\[\LenI(t_0,t)=\int_{t_0}^{t}\sigmap{p}\max(1,\infnorm{y(u)})^kdu\]
\end{definition}

\begin{lemma}[Relationship between $\IntI$ and $\LenI$]\label{lem:ineq_I_L}
For any $t\geqslant t_0$ in the domain of definition of $y$ and $\varepsilon\leqslant\frac{1}{4k}$,
\begin{equation*}
\IntI(t_0,t)\leqslant 2k\LenI(t_0,t)
\end{equation*}
\end{lemma}

\begin{proof}
\begin{align*}
\IntI(t_0,t)&=\int_{t_0}^{t}k\sigmap{p}\max(1,\varepsilon+\infnorm{y(u)})^{k-1}du\\
&\leqslant k\int_{t_0}^{t}\sigmap{p}(\varepsilon+\max(1,\infnorm{y(u)}))^kdu\\
&\leqslant k(1+\varepsilon)^k\int_{t_0}^{t}\sigmap{p}\max(1,\infnorm{y(u)})^kdu\\
&\leqslant ke^{k\log(1+\frac{1}{4k})}\LenI(t_0,t)\\
&\leqslant ke^{\frac{1}{4}}\LenI(t_0,t)
\end{align*}
\end{proof}

\begin{lemma}[Complexity of $\operatorname{SolvePIVPVariable}$]\label{lem:solve_pip_var_com}
The complexity of $\operatorname{SolvePIVPVariable}$ on input $(t_0,y_0,p,t,\varepsilon,I)$ is bounded by:
\begin{equation*}
C_{arith}=\poly\left(\degp{p}^d,I,\log\LenI(t_0,t),\log\infnorm{y_0},-\log\varepsilon\right)
\end{equation*}
\begin{equation*}
C_{\Q,bit}=\poly\left(\degp{p},I,\log\LenI(t_0,t),\log\infnorm{y_0},\log\sigmap{p},-\log\varepsilon\right)^d
\end{equation*}
\end{lemma}

\begin{proof}
It is clear that what makes up most of the complexity of the algorithm are the calls to
$\operatorname{ComputeTaylor}$. More precisely, let $C$ be the complexity of the algorithm, then:
\begin{equation*}
C_x=\bigO{\sum_{i=0}^{n-1}\operatorname{TL}_x(d,p,\tilde{y}_i,\omega_i,\mu_i,\delta t_i)}
\end{equation*}
where $x\in\{arith,bit\}$ and $\operatorname{TL}_{arith}$ and $\operatorname{TL}_{bit}$ 
are the arithmetic and bit complexity of computing Taylor series. In the case of the
rational bit complexity, we will discuss the complexity of the comparison as well as the
size of the rational numbers at the end of the proof.
In \secref{sec:ode_taylor} and \eqref{eq:tl_arith},\eqref{eq:tl_bit} precisely, we explained that one can show that
\begin{align*}
\operatorname{TL}_{arith}&=\softO{\omega \degp{p}^d+(d\omega)^a)}\\
\operatorname{TL}_{bit}&=\bigO{\poly((\degp{p}\omega)^d,\log\max(1,t)\sigmap{p}\max(1,\infnorm{y_0}),-\log\mu)}
\end{align*}
Recalling that $k=\degp{p}$ and that all time steps are lower than $1$, we get
\begin{align*}
\operatorname{TL}_{arith}&=\softO{k^d(d\omega)^a}\\
\operatorname{TL}_{bit}&=\bigO{\poly((k\omega)^d,\log \sigmap{p}\max(1,\infnorm{y_0}),-\log\mu)}
\end{align*}
Consequently, using \eqref{eq:choice_omega_mu_eps_eta},
\begin{align*}
C_{arith}&=\softO{\sum_{i=0}^{n-1}k^dd^a\left(\log_2\frac{6Ne^I\max(1,\infnorm{\tilde{y}_i})}{\varepsilon}\right)^a}\\
C_{\Q,bit}&=\bigO{\sum_{i=0}^{n-1}\poly\left(\begin{array}{l}
    k^d,\left(\log_2\frac{6Ne^I\max(1,\infnorm{\tilde{y}_i})}{\varepsilon}\right)^d,\\
    \log(\sigmap{p}\max(1,\infnorm{\tilde{y}_i})),-\log\frac{\varepsilon}{3Ne^I}\end{array}\right)}
\end{align*}
But we know that
\begin{align*}\label{eq:bound_tilde_y}
\infnorm{\tilde{y}_i}&\leqslant\varepsilon+\infnorm{y(t_i)}\\
&\leqslant\varepsilon+\infnorm{y_0+\int_{t_0}^{t_i}p(y(u))du}\\
&\leqslant\varepsilon+\infnorm{y_0}+\int_{t_0}^{t_i}\infnorm{p(y(u))}du\\
&\leqslant\varepsilon+\infnorm{y_0}+\int_{t_0}^{t_i}\sigmap{p}\max(1,\infnorm{y(u)})^kdu\\
&\leqslant\varepsilon+\infnorm{y_0}+\LenI(t_0,t_i)\\
\max(1,\infnorm{\tilde{y}_i})&\numberthis\leqslant1+\infnorm{y_0}+\LenI(t_0,t)\\
\end{align*}
Using that $\varepsilon\leqslant\frac{1}{4k}$, we also have:
\begin{equation}\label{eq:bound_N_I_applied}
\begin{aligned}
N&=1+\frac{2I}{\lambda}=1+2I\left(1+\frac{k}{1-2k\varepsilon}\right)\\
&\leqslant1+2I(1+2k)
\end{aligned}
\end{equation}
Which gives using that $n\leqslant N$ and that $a\leqslant3$,
\begin{align*}
C_{arith}&=\softO{\sum_{i=0}^{n-1}k^dd^a\left(\log_2\frac{\poly(e^I,k,\infnorm{y_0},\LenI(t_0,t))}{\varepsilon}\right)^a}\\
&=\poly\left(k^d,I,k,\log\left(\poly\left(e^I,k,\infnorm{y_0},\LenI(t_0,t),\frac{1}{\varepsilon}\right)\right)\right)\\
&=\poly\left(k^d,I,\log\LenI(t_0,t),\log\infnorm{y_0},-\log\varepsilon\right)\\
\end{align*}
And similarly:
\begin{align*}
C_{\Q,bit}&=\bigO{\sum_{i=0}^{n-1}\poly\left(k^d,
    \left(\log\frac{\poly(e^I,k,\infnorm{y_0},\LenI(t_0,t))}{\varepsilon}\right)^d,\log\sigmap{p}\right)}\\
    &\leqslant\poly\big(k,
    I,\log\infnorm{y_0},\log\LenI(t_0,t),\log\sigmap{p},-\log\varepsilon\big)^d
\end{align*}
In the case of rational bit complexity, we need to argue that all other operations
in the loop are polytime computable. This boils down to two facts: all rational numbers
must have polynomial size in the input, and all comparisons must be performed quickly.

The first point is nearly immediate: all assignments before the loop are made up
of addition, subtraction, multiplication and division which are all polynomial time
computable and there is a finite number of them. In the loop, we need to have a bound on
the size of $\tilde{y}$. There are two aspects to this bound: first, we need to ensure the number
of digits remains controlled, this is achieved by limiting the output precision of ComputeTaylor
to $-\log\mu$ digits. Second, we need to ensure that the numbers do not become too large,
and this comes from \eqref{eq:bound_tilde_y}.
In other words, using that $\lambda\leqslant1+2k$ and except maybe for the last iteration:
\begin{align*}
\Rsize{\tilde{y}}&\leqslant\poly(\log\max(1,\infnorm{\tilde{y}}),-\log\mu)\\
    &\leqslant\poly\left(\log\max(1,\infnorm{\tilde{y}}),-\log\frac{\varepsilon e^{-I}}{3N}\right)\\
    &\leqslant\poly\left(\log\max(1,\infnorm{\tilde{y}}),-\log\varepsilon,I,\log N\right)\\
    &\leqslant\poly\left(\log(1+\infnorm{y_0}+\LenI(t_0,t)),-\log\varepsilon,I,\log k\right)\tag{Use \eqref{eq:bound_N_I_applied}}\\
    &\leqslant\poly(\log\infnorm{y_0},\log\LenI(t_0,t),-\log\varepsilon,I,k)
\end{align*}
It follows immediately that $\beta$ and $u$ have polynomial size. If the last iteration
sets $\delta$ to $t-u$ then it follows that $\delta$ has polynomial size because
$t$ and $u$ have polynomial size. It follows that the comparison between $u$ and $t$
is polytime and we are left with showing that the comparison between $I$ and $3((i-1)\lambda+\beta)$
is polytime computable. The argument is very similar: $I$ has polynomial size
by definition, we already saw that $\beta$ has polynomial size, $i\leqslant N$ which
we already argued has polynomial size and
$\lambda\leqslant1+2k$ so $\lambda$ also has polynomial size. It follows that the test
can be performed in polytime because this is a comparison between two polynomial size
rational numbers.
\end{proof}

The algorithm of the previous section depends on an ``hint'' $I$ given as input by the user. 
This isn't certainly a desirable feature, since the algorithm is only guaranteed to terminate (with a correct answer on that case) if 
\[
I\geqslant\int_{t_0}^{t}k\sigmap{p}(1+\varepsilon+\infnorm{y(u)})^{k-1}du
\]
but the user has usually no way of estimating the right-hand side of this inequality (the problem is that it requires some knowledge
about the solution $y$ which we are trying to compute).

However we know that if the hint $I$ is large enough, then the algorithm will succeed in returning a result. Furthermore if it succeeds, the result
is correct. A very natural way of solving this problem is to repeatedly try
for larger values of the hint until the algorithm succeeds. We are guaranteed that
this will eventually happen when the hint $I$  reaches the theoretical bound given by Lemma \ref{lem:solve_pivp_taylor_correct} (although it can
stop much earlier in many cases). By choosing a very simple update strategy of the hint (double its value on each failure),
it is possible to see that this process does not cost significantly more than if we already had the hint.
In the case of the real bit complexity, things get more complicated because we have to approximate
all inputs with rational numbers to call the previous algorithm.

\begin{algorithm}
\begin{algorithmic}[1]
\Require{$t_0\in\R$ the initial time}
\Require{$y_0\in\R^d$ the initial condition}
\Require{$p\in\R^d[\R^d]$ polynomial of the PIVP}
\Require{$t\in\R$ the time step}
\Require{$\varepsilon\in\R$ the accuracy requested}
\Function{SolvePIVPEx}{$t_0,y_0,p,t,\varepsilon$}
\State $I\algassign1/2$\;
\Repeat
    \State $I\algassign2I$
    \State $x\algassign \Call{SolvePIVPVariable}{t_0,y_0,p,t,\varepsilon,I}$
\Until{$x\neq\bot$}
\State\Return{$x$}
\EndFunction
\end{algorithmic}
\caption{PIVP Solving algorithm: $\operatorname{SolvePIVPEx}$\label{alg:solve_pivp_x_boot}}
\end{algorithm}

\begin{theorem}[Complexity and correctness of $\operatorname{SolvePIVPEx}$]\label{th:solve_pivp_ex_taylor}
Let $t_0,t\in\R$, $\varepsilon>0$, and assume that $y$ satisfies \eqref{eq:pivp} over $[t_0,t]$. Let
\[x=\operatorname{SolvePIVPEx}(t_0,y_0,p,t,\varepsilon)\]
where  $\operatorname{SolvePIVPEx}$ is the Algorithm \ref{alg:solve_pivp_x_boot}.
Then $\infnorm{x-y(t)}\leqslant\varepsilon$ and the algorithm has complexity:
\begin{equation*}
C_{arith}=\poly\left(\degp{p}^d,\LenI(t_0,t),\log\infnorm{y_0},-\log\varepsilon\right)
\end{equation*}
\begin{equation*}
C_{\Q,bit}=\poly\left(\degp{p},\LenI(t_0,t),\log\infnorm{y_0},\log\sigmap{p},-\log\varepsilon\right)^d
\end{equation*}
where $\LenI(t_0,t)=\int_{t_0}^{t}\sigmap{p}\max(1,\infnorm{y(u)})^kdu$.
\end{theorem}

\begin{proof}
By \lemref{lem:solve_pivp_taylor_correct},
we know that the algorithm succeeds whenever $I\geqslant6\IntI(t_0,t)$.
Thus when the $I$ in the algorithm is greater than this bound, the loop must stop. 
Recall that the value of $I$ at the $i^{th}$ iteration is $2^i$ ($i$ starts at 0). Let $q$
be the number of iterations of the algorithm: it stops with $I=2^{q-1}$. Then:
\[6\cdot 2^{q-2}\leqslant6\IntI(t_0,t)\]
Indeed, if it wasn't the case, the algorithm would have stop one iteration earlier.
Using \lemref{lem:ineq_I_L} we get:
\[2^{q-1}=\bigO{k\LenI(t_0,t)}\qquad q=\bigO{\log(k\LenI(t_0,t))}\]
Now apply \lemref{lem:solve_pip_var_com} to get that the final complexity $C$ is bounded by:
\begin{align*}
C_{arith}&=\bigO{\sum_{i=0}^{q-1} \poly(k^d,2^i,\log\LenI(t_0,t),\log\infnorm{y_0},-\log\varepsilon)}\\
&=\bigO{q\poly(k^d,2^{q-1},\log\LenI(t_0,t),\log\infnorm{y_0},-\log\varepsilon)}\\
&=\poly(k^d,\LenI(t_0,t),\log\infnorm{y_0},-\log\varepsilon)
\end{align*}
Similarly:
\begin{align*}
C_{\Q,bit}&=\bigO{\sum_{i=0}^{q-1} \poly(k,2^i,\log\LenI(t_0,t),\log\sigmap{p},\log\infnorm{y_0},-\log\varepsilon)^d}\\
&=\bigO{q\poly(k,2^{q-1},\log\LenI(t_0,t),\log\infnorm{y_0},\log\sigmap{p},-\log\varepsilon)^d}\\
&=\poly(k,\LenI(t_0,t),\log\infnorm{y_0},\log\sigmap{p},-\log\varepsilon)^d
\end{align*}
\end{proof}

\section{Conclusion and future work}\label{sec:solve:conclusion}

In this paper we presented a method which allows us to solve a polynomial
ordinary differential equation over an unbounded time with an arbitrary accuracy. Moreover
our method is guaranteed to produce a result which has a certain accuracy, where the accuracy
is also provided as an input to our method.

We analyzed the method and established rigorous bounds on the time it needs to output a result.
In this manner we were able to determine the computational complexity of solving polynomial differential equations over unbounded domains
and identified the length of the solution curve as the right parameter to measure the complexity
of the algorithm.

Our work suffers from several limitations which we plan to investigate in the future.
First, from a practical perspective, a more precise complexity bound would be useful.
For example, is the complexity dependence in the accuracy linear, or quadratic, or more ?
Second, the quantity $\LenI(t_0,t)$ is an overapproximation to the actual length of the curve
and can sometimes be very far from it. In particular, as soon as the length of the curve
grows sublinearly with time, a big gap exists between the two. This is easily seen
on examples such as $y(t)=e^{-t}$ which has bounded length, and $y(t)=\log t$ which
has extremely slow growing length.

\textbf{Acknowledgments.} Daniel Gra\c{c}a was partially supported by
\emph{Funda\c{c}\~{a}o para a Ci\^{e}ncia e a Tecnologia} and EU FEDER
POCTI/POCI via SQIG - Instituto de Telecomunica\c{c}\~{o}es through the FCT project UID/EEA/50008/2013.

The authors would like to thank the anonymous reviewers for their helpful and constructive
comments that greatly contributed to improving the final version of the paper.

\bibliography{ContComp,bournez}
\bibliographystyle{elsarticle-harv}

\end{document}